\documentclass[11pt]{article}

\usepackage{cite}
\usepackage{graphics}
\usepackage{subfigure}
\usepackage{geometry}
\usepackage{amsmath}
\allowdisplaybreaks[1]

\usepackage{amsthm}

\newtheorem{theorem}{Theorem}
\newtheorem{lemma}[theorem]{Lemma}
\newtheorem{claim}[theorem]{Claim}

\makeatletter
\def\@endtheorem{\endtrivlist}
\makeatother

\newcommand{\psv}[1]{\mathrm{PSV}(#1)}
\newcommand{\plv}[1]{\mathrm{PLV}(#1)}
\newcommand{\nsv}[1]{\mathrm{NSV}(#1)}
\newcommand{\nlv}[1]{\mathrm{NLV}(#1)}

\newcommand{\minheap}[1]{\mathrm{Min}(#1)}
\newcommand{\maxheap}[1]{\mathrm{Max}(#1)}
\newcommand{\cminheap}[1]{\mathrm{cMin}(#1)}
\newcommand{\cmaxheap}[1]{\mathrm{cMax}(#1)}

\newcommand{\bad}{\mathrm{bad}}
\newcommand{\neutral}{\mathrm{neutral}}
\newcommand{\goodbad}{\mathrm{good/bad}}

\begin{document}

\title{The effective entropy of next/previous larger/smaller value queries}
\author{Dekel Tsur%
\thanks{Department of Computer Science, Ben-Gurion University of the Negev.
Email: \texttt{dekelts@cs.bgu.ac.il}}}
\date{}
\maketitle

\begin{abstract}
We study the problem of storing the minimum number of bits required
to answer next/previous larger/smaller value queries on an array $A$ of $n$
numbers, without storing $A$.
We show that these queries can be answered by storing at most $3.701 n$ bits.
Our result improves the result of Jo and Satti [TCS 2016] that gives an
upper bound of $4.088n$ bits for this problem.
\end{abstract}


\section{Introduction}
A recent research area in data structures is designing data structures in
the encoding model~\cite{BrodalBD213,BrodalDR2012space,DavoodiNRR14,
Fischer11,GagieMV17,GawrychowskiN15,GolinIKRSS16,GrossiINRS17,JayapaulJRRS16,
JoLS16,JoS16}.
In this model, the goal is to design a data structure for answering queries
on some object $A$, without storing $A$.
The space complexity of the data structure should be close to the minimum space
required in order to answer the queries, without storing $A$.
The minimum space required to answer the queries is called the
\emph{effective entropy} of the problem.

Let $A$ be an array of $n$ numbers.
Consider the following four queries on $A$.
\begin{itemize}
\item
Previous smaller value ($\psv{i}$):
Given $i$, return $\max (\{ j \colon j<i, A[j] < A[i] \}\cup \{0\})$.
\item
Previous larger value ($\plv{i}$):
Given $i$, return $\max (\{ j \colon j<i, A[j] > A[i] \} \cup \{0\})$.
\item
Next smaller value ($\nsv{i}$):
Given $i$, return $\min (\{ j \colon j>i, A[j] < A[i] \} \cup \{n+1\})$.
\item
Next larger value ($\nlv{i}$):
Given $i$, return $\min (\{ j \colon j>i, A[j] > A[i] \}\cup \{n+1\})$.
\end{itemize}

The effective entropy of answering one type of queries from the four types above
is $2n-\Theta(\log n)$ bits.
If the problem is to support more than one type of queries, the effective
entropy becomes larger.
Fischer~\cite{Fischer11} showed that the effective entropy of answering both
PSV and NSV queries is $\log(3+2\sqrt{2})\cdot n-\Theta(\log n) < 2.544n$ bits.
Gawrychowski and Nicholson~\cite{GawrychowskiN15} showed that the effective
entropy of answering both PSV and PLV queries is at most $3n$ bits and at least
$3n-\Theta(\log n)$ bits.
For each of the two problems above, it is possible to build a data structure
that answers queries in constant time and with space that is equal to the
effective entropy plus $o(n)$ bits~\cite{Fischer11,GawrychowskiN15}.

Jo and Satti~\cite{JoS16} studied the problem of answering all four
queries. They showed that the effective entropy is at most $4n+o(n)$ bits
on arrays with no consecutive equal elements,
and at most $\log_2 17\cdot n +o(n) < 4.088n$ bits on general arrays.
They also showed that it is possible to build data structures that answer
queries in constant time with space complexity $4n+o(n)$ bits
on arrays with no consecutive equal elements,
and $4.585n$ bits for general arrays.
In this paper we improve the results of Jo and Satti.
We show that the effective entropy for answering all four queries is at most
$(2+\log 3)n +o(n) < 3.585n$ bits on array with no consecutive equal elements,
and at most $\log 13 \cdot n+o(n) < 3.701 n$ bits on general arrays.

\section{Preleminaries}
\subsection{Encoding PSV queries}\label{sec:psv}
The \emph{2d-min heap}~\cite{FischerH11} of an array $A$ of size $n$,
denoted $\minheap{A}$,
is an ordinal tree with nodes $0,\ldots,n$.
For every $i > 0$, the parent of node $i$ is $\psv{i}$.
The children of a node $i$ are ordered in increasing order of their names.
See Figure~\ref{fig:minheap} for an example.
Note that the preorder of the nodes is $0,\ldots,n$.
Therefore, the tree $\minheap{A}$ can be encoded by just storing its
topology.
Since $\minheap{A}$ is an ordinal tree with $n+1$ nodes, it follows that
the effective entropy of PSV queries is at most
$\lceil\log C_{n+1}\rceil = 2n-\Theta(\log n)$,
where $C_n = \frac{1}{n+1}\binom{2n}{n}$ is the number of ordinal trees with
$n$ nodes.
We also define the 2d-max heap, denoted $\maxheap{A}$, to be an ordinal tree
in which the parent of node $i>0$ is $\plv{i}$.

\begin{figure}
\centering
\subfigure{\includegraphics{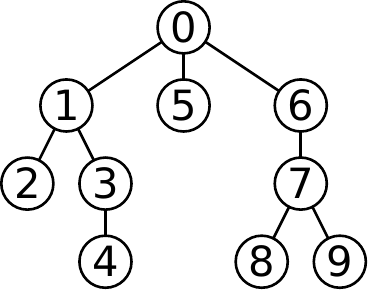}}

\begin{tabular}{|c|c|c|c|c|c|c|c|c|c|}
\hline
$i$    & 1 & 2 & 3 & 4 & 5 & 6 & 7 & 8 & 9\\
$A[i]$ & 3 & 8 & 5 & 6 & 3 & 2 & 7 & 10 & 9\\
\hline
\end{tabular}
\caption{An example of a 2d-min heap of an array $A$.}
\label{fig:minheap}
\end{figure}


%

\subsection{Encoding PSV/NSV queries}\label{sec:psv-nsv}
In order to encode both PSV and NSV queries, Fischer~\cite{Fischer11} defined
the \emph{colored 2d-min heap} of an array $A$, denoted $\cminheap{A}$,
to be the tree $\minheap{A}$ with the following coloring of its nodes.
If node $i$ has right siblings and $A[i] \neq A[j]$, where $j$ is the immediately
right sibling of $i$, then $i$ is colored red.
Otherwise, $i$ is colored blue.
See Figure~\ref{fig:cminheap} for an example.
Fischer showed that if $\cminheap{A}$ is known, both PSV and NSV queries on $A$
can be answered without storing $A$.
Since $\cminheap{A}$ is a Schr\"oder tree, it follows that the effective entropy
of PSV/NSV queries is at most $\log(3+2\sqrt{2})\cdot n-\Theta(\log n)$
bits~\cite{MerliniSV04}.

\begin{figure}
\centering
\subfigure{\includegraphics{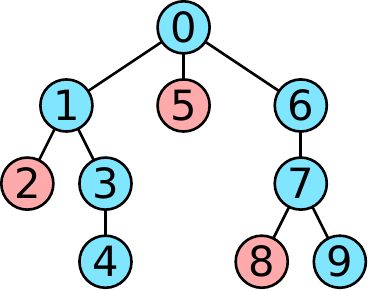}}

\begin{tabular}{|c|c|c|c|c|c|c|c|c|c|}
\hline
$i$    & 1 & 2 & 3 & 4 & 5 & 6 & 7 & 8 & 9\\
$A[i]$ & 3 & 8 & 5 & 6 & 3 & 2 & 7 & 10 & 9\\
\hline
\end{tabular}
\caption{An example of a colored 2d-min heap of an array $A$.}
\label{fig:cminheap}
\end{figure}

\subsection{Encoding PSV/PLV queries}\label{sec:psv-plv}
To answer both PSV and PLV queries, we can store both
$\minheap{A}$ and $\maxheap{A}$.
Gawrychowski and Nicholson\cite{GawrychowskiN15} showed that
for an array $A$ with no consecutive equal elements,
$\minheap{A}$ and $\maxheap{A}$ can be encoded using $3n-1$ bits.
The proof of this result is based on the following claim.

\begin{claim}\label{clm:leaf-internal}
If $A$ has no consecutive equal elements, then for every $0 < i < n$,
$i$ is a leaf in $\minheap{A}$ if and only if
$i$ is an internal node of $\maxheap{A}$
\end{claim}
\begin{proof}
Since $A$ has no consecutive equal elements, we have that
either $A[i] < A[i+1]$ or $A[i] > A[i+1]$.
In the former case, $\psv{i+1} = i$. Therefore, $i+1$ is a child of $i$ in
$\minheap{A}$, and thus $i$ is an internal node in $\minheap{A}$.
Additionally, $\plv{i+1} < i$, so  $i+1$ is not a child of $i$ in $\maxheap{A}$.
Since the names of the nodes of $\maxheap{A}$ are according to their ranks in
the preorder, the descendants of an internal node $j$ are $j+1,\ldots,j'$ for
some $j'$.
Since $i+1$ is not a descendant of $i$ in $\maxheap{A}$, it follows that
$i$ is a leaf in $\maxheap{A}$.

The proof for the case when $A[i] > A[i+1]$ is analogous and thus omitted.
\end{proof}

The encoding $\minheap{A}$ and $\maxheap{A}$ consists of three binary strings
$U$, $T_{\minheap{A}}$, and $T_{\maxheap{A}}$.
The string $U$ is a string of length $n-1$ in which $U[i] = 1$ if and only if
$i$ is a leaf in $\minheap{A}$.
The strings $T_{\minheap{A}}$ and $T_{\maxheap{A}}$ are defined as follows.
Start with empty strings $T_{\minheap{A}}$ and $T_{\maxheap{A}}$.
Then, for $i=0,1,\ldots,n-1$, if $i = 0$ or $U[i] = 0$
(namely, $i$ is an internal node in $\minheap{A}$),
append the string $1^{d_i-1}0$ to $T_{\minheap{A}}$,
where $d_i$ is the number of children of node~$i$ in $\minheap{A}$.
Additionally, if $i = 0$ or $U[i] = 1$,
append the string $1^{d'_i-1}0$ to $T_{\maxheap{A}}$,
where $d'_i$ is the number of children of node~$i$ in $\maxheap{A}$.
See Figure~\ref{fig:minheap+maxheap} for an example.

By Claim~\ref{clm:leaf-internal}, the trees $\minheap{A}$ and $\maxheap{A}$
can be reconstructed from the strings $U,T_{\minheap{A}},T_{\maxheap{A}}$.
It is easy to show that $|T_{\minheap{A}}|+|T_{\maxheap{A}}|=2n$.
Therefore, the total size of this encoding is $3n-1$ bits.

\begin{figure}
\centering
\subfigure[$\minheap{A}$]{\includegraphics{figs/minheap}}
\subfigure[$\maxheap{A}$]{\includegraphics{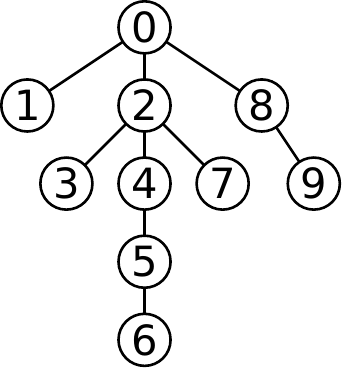}}

\begin{tabular}{|c|c|c|c|c|c|c|c|c|c|c|}
\hline
$i$    & 0 & 1 & 2 & 3 & 4 & 5 & 6 & 7 & 8 & 9\\
$A[i]$ &   & 3 & 8 & 5 & 6 & 3 & 2 & 7 & 10 & 9\\
$U$    &   & 0 & 1 & 0 & 1 & 1 & 0 & 0 & 1 & \\
$T_{\minheap{A}}$ & 110 & 10 &  & 0 &  &  & 0 & 10 &  & \\
$T_{\maxheap{A}}$ & 110 &  & 110 &  & 0 & 0 &  &  & 0 & \\
\hline
\end{tabular}
\caption{An example of the trees $\minheap{A}$ and $\maxheap{A}$
of an array $A$ and the encoding of these trees using $3n-1$ bits.}
\label{fig:minheap+maxheap}
\end{figure}

\section{Arrays with no consecutive equal elements}
\begin{theorem}\label{thm:no-consecutive}
The effective entropy of PSV/PLV/NSV/NLV queries on arrays with no consecutive
equal elements is at most $(2+\log 3)n +o(n) < 3.585n$.
\end{theorem}
\begin{proof}
To answer PSV/PLV/NSV/NLV queries, it suffices to encode the trees
$\cminheap{A}$ and $\cmaxheap{A}$. We will show that this can be done
using $(2+\log 3)n + o(n)$ bits.
To reduce the size of the encoding, we use the following claim.
\begin{claim}\label{clm:red}
If $i$ is a leaf in $\cminheap{A}$ (resp., $\cmaxheap{A}$) and $i$ has
right siblings then $i$ is red in $\cminheap{A}$ (resp., $\cmaxheap{A}$).
\end{claim}
\begin{proof}
Since $i$ is a leaf, the immediate right sibling of $i$ is $i+1$.
Due to the assumption that $A$ has no consecutive equal elements we have
that $A[i] \neq A[i+1]$. Therefore, $i$ is red.
\end{proof}
We say that an index $0 < i < n$ is \emph{good} if $i$ does not have
right siblings in $\cminheap{A}$ and in $\cmaxheap{A}$.
We say that $0 < i < n$ is \emph{bad} if $i$ has right siblings in
$\cminheap{A}$ and in $\cmaxheap{A}$.
If $0 < i < n$ is not good or bad we say that $i$ is \emph{neutral}.
For an index $0<i<n$, the \emph{relevant tree of $i$} is the tree from
$\cminheap{A},\cmaxheap{A}$ in which $i$ is an internal node
(note that by Claim~\ref{clm:leaf-internal} there is exactly one tree in which
$i$ is an internal node).

In the following, we encode the color red by $0$ and the color blue by $1$.
The encoding of $\cminheap{A},\cmaxheap{A}$ consists of the following strings.
\begin{itemize}
\item
The strings $T_{\minheap{A}}$ and $T_{\maxheap{A}}$ (these strings
were defined in Section~\ref{sec:psv-plv}).
\item
A binary string $U_\goodbad$ that is obtained by concatenating the
characters $U[i]$ for every $i$ which is either good or bad
(the string $U$ was defined in Section~\ref{sec:psv-plv}).
\item
A binary string $V_\bad$ obtained by concatenating the color of $i$ in
the relevant tree of $i$ for every bad $i$.
\item
A ternary string $V_\neutral$ obtained by concatenating a character $c_i$
for every neutral index $i$.
If $i$ does not have right siblings in its
relevant tree, $c_i = 2$.
Otherwise, $c_i$ is the color of $i$ in the relevant tree of $i$.
\end{itemize}
See Figure~\ref{fig:cminheap+cmaxheap} for an example.

\begin{figure}
\centering
\subfigure[$\cminheap{A}$]{\includegraphics{figs/cminheap}}
\subfigure[$\cmaxheap{A}$]{\includegraphics{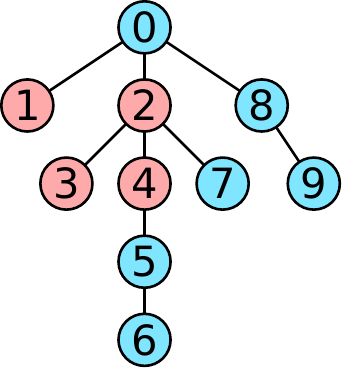}}

\begin{tabular}{|c|c|c|c|c|c|c|c|c|c|c|}
\hline
$i$    & 0     & 1 & 2 & 3 & 4 & 5 & 6 & 7 & 8 & 9\\
$A[i]$ &       & 3 & 8 & 5 & 6 & 3 & 2 & 7 & 10 & 9\\
$T_{\minheap{A}}$ & 110 & 10 &  & 0 &  &  & 0 & 10 &  & \\
$T_{\maxheap{A}}$ & 110 &  & 110 &  & 0 & 0 &  &  & 0 & \\
$U_\goodbad$ & & 0 & 1 &   &   &   & 0 & 0 &   & \\
$V_\bad$     & & 1 & 0 &   &   &   &   &   &   & \\
$V_\neutral$ & &   &   & 2 & 0 & 2 &   &   & 2 & \\
\hline
\end{tabular}
\caption{An example of the trees $\cminheap{A}$ and $\cmaxheap{A}$
of an array $A$ and the encoding of these trees.
The good indices are 6,7 and the bad indices are 1,2.}
\label{fig:cminheap+cmaxheap}
\end{figure}

We now show that given the string $T_{\minheap{A}}$, $T_{\maxheap{A}}$,
$U_\goodbad$, $V_\bad$, and $V_\neutral$ we can reconstruct the trees
$\cminheap{A}$ and $\cmaxheap{A}$.
We initialize two trees $T_1,T_2$ to contain node $0$.
At the end of the following algorithm,
$T_1 = \cminheap{A}$ and $T_2 = \cmaxheap{A}$.
By reading the prefixes of $T_{\minheap{A}}$ and $T_{\maxheap{A}}$ until the
first zero in each string, we know the degrees of node $0$ in
$\minheap{A}$ and in $\maxheap{A}$.
We now go over $i=1,\ldots,n$ and add node $i$ to the trees $T_1$ and $T_2$.
This is done as follows.
From the previous iterations of the algorithm,
we know the number of children of the nodes $0,\ldots,i-1$ in $\minheap{A}$
and in $\maxheap{A}$.
We make node $i$ the rightmost child of node $j$ in $T_1$ where $j$ is the
minimum integer such that the number of children of $j$ in $T_1$ is less than
the number of children of $j$ in $\minheap{A}$.
We also add node $i$ to the tree $T_2$ similarly.
We then find which tree is the relevant tree of $i$.
After we know the relevant tree of $i$, we read unread characters from the
string $T_{\minheap{A}}$ or $T_{\maxheap{A}}$ that corresponds the relevant
tree of $i$ until reaching the first zero.
This gives us the number of children of $i$ in its relevant tree.
Moreover, by Claim~\ref{clm:leaf-internal},
$i$ does not have children in the non-relevant tree.
Finally, we find the color of $i$ in $\minheap{A}$ and $\maxheap{A}$.

Since we know the number of children of the parent of $i$ in $\minheap{A}$,
we know whether node $i$ has right siblings in $\minheap{A}$
($i$ has right siblings if and only if the number of children of the parent
of $i$ in $T_1$ is less than the number of children of the parent of $i$ in
$\minheap{A}$).
Similarly, we know whether $i$ has a right sibling in $\maxheap{A}$.
Therefore, we know whether $i$ is good, bad, or neutral.

If $i$ is good, reading the next unread character from $U_\goodbad$ gives us
the relevant tree of $i$.
Since $i$ has no right siblings in $\cminheap{A}$ and in $\cmaxheap{A}$,
the color of $i$ is blue in both trees.

If $i$ is bad, reading the next unread character from $U_\goodbad$ gives us
the relevant tree of $i$. We also read the next unread character from $V_\bad$
to know the color of $i$ in the relevant tree. The color of $i$ in the
non-relevant tree is red by Claim~\ref{clm:red}.

Finally, if $i$ is neutral, we read the next unread character from $V_\neutral$
and let $c$ be this character.
Suppose without loss of generality that $i$ does not have right siblings in
$\cminheap{A}$ and it has right siblings in $\cmaxheap{A}$.
The color of $i$ in $\cminheap{A}$ is blue (since $i$ does not have right
siblings in $\cminheap{A}$).
If $c = 2$ then the relevant tree of $i$ is $\cminheap{A}$ and
the color of $i$ in $\cmaxheap{A}$ is red (by Claim~\ref{clm:red}).
Otherwise, the relevant tree of $i$ is $\cmaxheap{A}$ and
the color of $i$ in $\cmaxheap{A}$ is red if $c=0$ and blue if $c=1$.

We now analyze the size of the encoding. We need the following lemma.
\begin{lemma}\label{lem:good-bad}
The number of good indices is equal to the number of bad indices.
\end{lemma}
\begin{proof}
For the purpose of the proof we also define the index $n$ to be a good index.
Thus, we now need to show that the number of good indices is equal to the
number of bad indices plus one.
We prove this claim using induction on $n$.
The base of the induction, $n = 1$, is true since in this case there is one
good index and no bad indices.

Now suppose that $n > 1$.
Suppose without loss of generality that $A[n-1] > A[n]$.
Then, node $n$ is the only child of $n-1$ in $\cmaxheap{A}$.
Moreover, node $n$ is not a child of $n-1$ in $\cminheap{A}$.
Let $j$ be the parent of $n$ in $\cminheap{A}$.
Since the names of the nodes are according to their ranks in the preorder,
we have that nodes $j+1,\ldots,n-1$ are descendants of $j$ in $\cminheap{A}$,
and therefore node $j+1$ is a child of $j$ in $\cminheap{A}$.
Therefore, $n$ has left siblings in $\cminheap{A}$.
Let $k$ be the immediate left sibling of $n$ in $\cminheap{A}$.

Let $A'$ be the array obtained by taking the first $n-1$ elements of $A$.
We have that the trees $\minheap{A'}$ and $\maxheap{A'}$ are obtained by
deleting node $n$ from $\minheap{A}$ and $\maxheap{A}$, respectively.

Since $n$ is the single child of $n-1$ in $\maxheap{A}$, we conclude that
a node $i$ has right siblings in $\maxheap{A'}$ if and only if
$i$ has right siblings in $\maxheap{A}$.
Since $n$ is the immediate right sibling of $k$ in $\minheap{A}$, we have that
a node $i\neq k$ has right siblings in $\minheap{A'}$ if and only if
$i$ has right siblings in $\minheap{A}$.
Moreover, $k$ has right siblings in $\minheap{A}$ but not in $\minheap{A'}$.
It follows that
\begin{itemize}
\item $n$ is a good index with respect to $A$, but not with respect to $A'$.
\item Either $k$ is a good index with respect to $A'$ and neutral with respect
to $A$ (if $k$ does not have right siblings in $\maxheap{A'}$),
or $k$ is a neutral index with respect to $A'$ and bad with respect to $A$.
\item For every $i \neq k,n$, $i$ is good (resp., bad) with respect to $A'$
if and only if $i$ is good (resp., bad) with respect to $A'$.
\end{itemize}
It follows that the difference between the number of good indices and bad
indices with respect to $A$ is equal to the difference of these numbers with
respect to $A'$.
By the induction hypothesis, the latter difference is $1$.
\end{proof}

Let $g$ be the number of good indices.
The combined size of $T_{\minheap{A}}$ and $T_{\maxheap{A}}$ is $2n$ bits.
The size of $U_\goodbad$ is $2g$ bits and the size of $V_\bad$ is $g$ bits.
The string $V_\neutral$ has length $n-1-2g$ so it can be encoded using
$\lceil(n-1-2g)\log 3\rceil$ bits.
We also need to store the lengths of the strings which requires
$O(\log n)$ bits.
The total size of the encoding is $2n+3g+(n-2g)\log 3+O(\log n)$ bits.
This expression is maximized when $g = 0$, and the theorem follows.
\end{proof}

\section{General arrays}
\begin{lemma}\label{lem:inequality}
For a constant $c > 0$ and integers $n,k$,
$c(n-k)+\log \binom{n}{k} \leq \log{(2^c+1)}\cdot n$.
\end{lemma}
\begin{proof}
Let $y = k/n$. We have that
\[c(n-k)+\log \binom{n}{k} = c(1-y)n+\log\binom{n}{yn}
\leq c(1-y)n + n\log\left(\frac{1}{y^y{(1-y)}^{1-y}}\right).
\]
Let $f(x) = c(1-x)+\log(\frac{1}{x^x{(1-x)}^{1-x}})$.
The derivative of $f$ is $-c+\log(1-x)-\log x$. Therefore, $f$ is maximized at
$x^*=1/(2^c+1)$ and $f(x^*) = \log(2^c+1)$.
\end{proof}

\begin{theorem}
The effective entropy of PSV/PLV/NSV/NLV queries is at most
$\log 13 \cdot n +o(n) < 3.701 n$.
\end{theorem}
\begin{proof}
As in~\cite{JoS16}, we define a binary string $C$ of length $n-1$ in which
$C[i] = 1$ if and only if $A[i] = A[i+1]$.
We also define an array $A'$ that is obtained from $A$ by deleting the
elements $A[i]$ for every $i$ such that $C[i] = 1$.
Let $k$ be the number of ones in $C$.
In order to answer PSV/PLV/NSV/NLV queries on $A$, it suffices to
store $C$ and information for answering PSV/PLV/NSV/NLV queries on $A'$.
By Theorem~\ref{thm:no-consecutive}, the latter can be done using
at most $(2+\log 3)(n-k)$ bits.
Moreover, $C$ can be stored using
$\lceil \log n\rceil+\lceil \log \binom{n}{k}\rceil$ bits.
The theorem now follows from Lemma~\ref{lem:inequality}.
\end{proof}

\bibliographystyle{plain}
\bibliography{ds}

\end{document}